\documentclass[11pt]{article}

\usepackage{amsmath}
\usepackage{amssymb}
\usepackage{amsthm}
\usepackage{fullpage}
\newif\ifpublic
\publictrue %%%%  Author's notes and TODO's suppressed
\ifpublic
\usepackage[disable]{todonotes}
\else
\usepackage[colorinlistoftodos]{todonotes}
\fi

\newcommand{\eps}{\varepsilon}
\newcommand{\mb}{\mathbf}
\newcommand{\tsfrac}[2]{{\textstyle\frac{#1}{#2}}}

\usepackage{comment}

\newtheorem{theorem}{Theorem}
\newtheorem{corollary}[theorem]{Corollary}
\newtheorem{lemma}[theorem]{Lemma}

\newtheorem{definition}[theorem]{Definition}

\newcommand{\prob}[2]{\mathop{\mathrm{Pr}}_{#1}[#2]}

\newcommand{\F}{\mathbb{F}}

\newcommand{\Maj}{\mathrm{Maj}}
\newcommand{\AC}{\mathrm{AC}}

\usepackage{enumerate}

\begin{document}

\title{Separation of AC$^0[\oplus]$ Formulas and Circuits}
\author{Benjamin Rossman\footnote{Supported by NSERC}\\ Depts.\ of Math and Computer Science\\ University of Toronto \and Srikanth Srinivasan\\ Department of Mathematics\\ IIT Bombay}
\maketitle{}

\begin{abstract}
This paper gives the first separation between the power of {\em formulas} and {\em circuits} of equal depth in the $\AC^0[\oplus]$ basis (unbounded fan-in AND, OR, NOT and MOD$_2$ gates). 
We show, for all $d(n) \le O(\frac{\log n}{\log\log n})$, that there exist {\em polynomial-size depth-$d$ circuits} that are not equivalent to {\em depth-$d$ formulas of size $\smash{n^{o(d)}}$} (moreover, this is optimal in that $\smash{n^{o(d)}}$ cannot be improved to $\smash{n^{O(d)}}$).
This result is obtained by a combination of new lower and upper bounds for {\em Approximate Majorities}, the class of Boolean functions $\{0,1\}^n \to \{0,1\}$ 
that agree with the Majority function on $3/4$ fraction of inputs.
\medskip

\textbf{AC$^{\textbf 0}$[$\pmb\oplus$] formula lower bound.}
We show that every depth-$d$ $\AC^0[\oplus]$ formula of size~$s$ has a {\em $1/8$-error polynomial approximation} over $\F_2$ of degree $O(\frac1d\log s)^{d-1}$. This strengthens a classic $O(\log s)^{d-1}$ degree approximation for \underline{circuits} due to Razborov \cite{Razborov}. Since the Majority function has approximate degree $\Theta(\sqrt n)$, this result implies an $\exp(\Omega(dn^{1/2(d-1)}))$ lower bound on the depth-$d$ $\AC^0[\oplus]$ formula size of all Approximate Majority functions for all $d(n) \le O(\log n)$.\medskip

\textbf{Monotone AC$^{\textbf 0}$ circuit upper bound.}
For all $d(n) \le O({\frac{\log n}{\log\log n}})$, we give a randomized construction of depth-$d$ monotone $\smash{\AC^0}$ circuits (without NOT or MOD$_2$ gates) of size $\smash{\exp(O(n^{1/2(d-1)}))}$ that compute an Approximate Majority function. This strengthens a construction of \underline{formulas} of size $\exp(O(dn^{1/2(d-1)}))$ due to Amano \cite{Amano-appxmaj}.
\end{abstract}

\section{Introduction}

The relative power of formulas versus circuits is one of the great mysteries in complexity \mbox{theory}. The central question in this area is whether NC$^1$ (the class of languages decidable by polynomial-size Boolean formulas) is a proper subclass of P/poly (the class of languages decidable by polynomial-size Boolean circuits). Despite decades of efforts, this question remains wide open.\footnote{In this paper we focus on non-uniform complexity classes. The question of uniform-NC$^1$ vs.\ P is wide open as well.}
In the meantime, there has been progress on analogues of the NC$^1$ vs.\ P/poly question in certain restricted settings. For instance, in the monotone basis (with AND and OR gates only), the power of polynomial-size formulas vs.\ circuits was separated by the classic lower bound of Karchmer and Wigderson \cite{karchmer1990monotone} (on the monotone formula size of st-Connectivity).

The bounded-depth setting is another natural venue for investigating the question of formula vs.\ circuits. Consider the elementary fact that every depth-$d$ circuit of size $s$ is equivalent to a depth-$d$ formula of size at most $s^{d-1}$, where we measure {\em size} by the number of gates. This observation is valid with respect to any basis (i.e.\ set of gate types). In particular, we may consider the $\AC^0$ basis (unbounded fan-in AND, OR, NOT gates) and the $\AC^0[\oplus]$ basis (unbounded fan-in MOD$_2$ gates in addition to AND, OR, NOT gates). With respect to either basis, there is a natural depth-$d$ analogue of the NC$^1$ vs.\ P/poly question (where $d = d(n)$ is a parameter that may depend on $n$), namely whether every language decidable by polynomial-size depth-$d$ circuits is decidable by depth-$d$ formulas of size $n^{o(d)}$ (i.e.\ better than the trivial $n^{O(d)}$ upper bound).

It is reasonable to expect that this question could be resolved in the sub-logarithmic depth regime ($d(n) \ll \log n$), given the 
powerful lower bound techniques against $\AC^0$ circuits (H{\aa}stad's Switching Lemma \cite{Hastad}) and $\AC^0[\oplus]$ circuits (the Polynomial Method of Razborov \cite{Razborov} and Smolensky \cite{Smolensky87}). However, because the standard way of applying these techniques does not distinguish between circuits and formulas, it is not clear how to prove quantitatively stronger lower bounds on formula size vis-a-vis circuit size of a given function. Recent work of Rossman \cite{rossman2015average} developed a new way of applying H{\aa}stad's Switching Lemma to $\AC^0$ formulas, in order to prove an $\exp(\Omega(dn^{1/(d-1)}))$ lower bound on the formula size of the Parity function for all $d \le O(\log n)$. Combined with the well-known $\exp(O(n^{1/(d-1)}))$ upper bound on the circuit size of Parity, this yields an asymptotically optimal separation in the power of depth-$d$ $\AC^0$ formulas vs. circuits for all $d(n) \le O(\frac{\log n}{\log\log n})$, as well as a super-polynomial separation for all $\omega(1) \le d(n) \le o(\log n)$.

In the present paper, we carry out a similar development for formulas vs.\ circuits in the $\AC^0[\oplus]$ basis, obtaining both an asymptotically optimal separation for all $d(n) \le O(\frac{\log n}{\log\log n})$ and a super-polynomial separation for all $\omega(1) \le d(n) \le o(\log n)$. Our target functions lie in the class of {\em Approximate Majorities}, here defined as Boolean functions $\{0,1\}^n \to \{0,1\}$ that approximate the Majority function on $3/4$ fraction of inputs. First, we show how to apply the Polynomial Method to obtain better parameters in the approximation of $\AC^0[\oplus]$ formulas by low-degree polynomials over $\F_2$. This leads to an $\exp(\Omega(dn^{1/2(d-1)}))$ lower bound on the $\AC^0[\oplus]$ formula size of all Approximate Majority functions. The other half of our formulas vs.\ circuits separation comes from an $\exp(O(n^{1/2(d-1)}))$ upper bound on the $\AC^0[\oplus]$ circuit size of some Approximate Majority function. In fact, this upper bound is realized by a randomized construction of monotone $\AC^0$ circuits (without NOT or MOD$_2$ gates). Together these upper and lower bound give our main result:

\begin{theorem}\
\begin{enumerate}[\normalfont(i)]
\item
For all $2 \le d(n) \le O(\frac{\log n}{\log \log n})$, there exist $\AC^0[\oplus]$ circuits \textup(in fact, monotone $\AC^0$ circuits\textup) of depth $d$ and size $\mathrm{poly}(n)$ that are not equivalent to any $\AC^0[\oplus]$ formulas of depth $d$ and size $n^{o(d)}$.
\item
For all $\omega(1) \le d(n) \le o(\log n)$, the class of languages decidable by polynomial-size depth-$d$ $\AC^0[\oplus]$ formulas is a proper subclass of the class of languages decidable by polynomial-size depth-$d$ $\AC^0[\oplus]$ circuits. 
\end{enumerate}
\end{theorem}

Separation (i) is asymptotically optimal, in view of the aforementioned simulation of $\mathrm{poly}(n)$-size depth-$d$ circuits by depth-$d$ formulas of size $n^{O(d)}$. Separation (ii)  resembles an analogue of NC$^1$ $\ne$ P/poly (or rather NC$^1$ $\ne$ AC$^1$) within the class $\AC^0[\oplus]$. In fact, extending separation (ii) from depth $o(\log n)$ to depth $\log n$ is equivalent the separation of NC$^1$ and AC$^1$.

\subsection{Proof outline}

\paragraph{Improved polynomial approximation.} The lower bound for $\AC^0[\oplus]$ formulas follows the general template due to Razborov~\cite{Razborov} on proving lower bounds for $\AC^0[\oplus]$ circuits using low-degree polynomials over $\F_2$. Razborov showed that for any Boolean function $f:\{0,1\}^n\rightarrow \{0,1\}$ that has an $\AC^0[\oplus]$ circuit of size $s$ and depth $d$, there is a randomized polynomial $\mb P$ of degree $O(\log s)^{d-1}$ that computes $f$ correctly on each input with probability $\frac{7}{8}$ (we call such polynomials $1/8$-error \emph{probabilistic polynomials}). By showing that some explicit Boolean function $f$ (e.g.\ the Majority function or the $\mathrm{MOD}_q$ function for $q$ odd) on $n$ variables does not have such an approximation of degree less than $\Omega(\sqrt{n})$~\cite{Razborov,Smolensky87,Smolensky93}, we get that any $\AC^0[\oplus]$ circuit of depth $d$ computing $f$ must have size $\exp(\Omega(n^{1/2(d-1)}))$.

In this paper, we improve the parameters of Razborov's polynomial approximation from above for $\AC^0[\oplus]$ \emph{formulas}. More precisely, for $\AC^0[\oplus]$ formulas of size $s$ and depth $d$, we are able to construct $1/8$-error probabilistic polynomials of degree $O(\frac{1}{d}\log s)^{d-1}$. (Since every depth-$d$ circuit of size $s$ is equivalent to a depth-$d$ formula of size at most $s^{d-1}$, this result implies Razborov's original theorem that $\AC^0[\oplus]$ circuits of size $s$ and depth $d$ have $1/8$-error probabilistic polynomials of degree $O(\log s)^{d-1}$.)

We illustrate the idea behind this improved polynomial approximation with the special case of a balanced formula (i.e. all gates have the same fan-in) of fan-in $t$ and depth $d$. Note that the size of the formula (number of gates) is $\Theta(t^{d-1})$ and hence it suffices in this case to show that it has a $1/8$-error probabilistic polynomial of degree $O(\log t)^{d-1}$. We construct the probabilistic polynomial inductively. Given a balanced formula $F$ of depth $d$ and fan-in $t$, let $F_1,\ldots,F_t$ be its subformulas of depth $d-1$. Inductively, each $F_i$ has a $1/8$-error probabilistic polynomial ${\mb P}_i$ of degree $O(\log t)^{d-2}$ and by a standard error-reduction~\cite{kopparty2012certifying}, it has a $(1/16t)$-error probabilistic polynomial of degree $O(\log t)^{d-1}$ (in particular, at any given input $x\in \{0,1\}^n$, the probability that \emph{there exists} an $i\in [t]$ such that ${\mb P}_i(x)\neq F_i(x)$ is at most $1/16$). Using Razborov's construction of a $1/16$-error probabilistic polynomial of degree $O(1)$ for the output gate of $F$ and composing this with the probabilistic polynomials ${\mb P}_i$, we get the result for balanced formulas. This idea can be extended to general (i.e. not necessarily balanced) formulas with a careful choice of the error parameter for each subformula $F_i$ to obtain the stronger polynomial approximation result. 

\paragraph{Improved formula lower bounds.} Combining the above approximation result with known lower bounds for polynomial approximation~\cite{Razborov,Smolensky87,Smolensky93}, we can already obtain stronger lower bounds for $\AC^0[\oplus]$ formulas than are known for $\AC^0[\oplus]$ circuits. For instance, it follows that any $\AC^0[\oplus]$ formula of depth $d$ computing the Majority function on $n$ variables must have size $\exp(\Omega(dn^{1/2(d-1)}))$ for all $d \le O(\log n)$, which is stronger than the corresponding circuit lower bound. Similarly stronger formula lower bounds also follow for the $\mathrm{MOD}_q$ function ($q$ odd). 

\paragraph{Separation between formulas and circuits.} However, the above improved lower bounds do not directly yield the claimed separation between $\AC^0[\oplus]$ formulas and circuits. This is because we do not have circuits computing (say) the Majority function of the required size. To be able to prove our result, we would need to show that the Majority function has $\AC^0[\oplus]$ circuits of depth $d$ and size $\exp(O(n^{1/2(d-1)}))$ (where the constant in the $O(\cdot)$ is independent of $d$). However, as far as we know, the strongest result in this direction~\cite{KPPY84majubd} only yields $\AC^0[\oplus]$ circuits of size greater than 
$\exp(\Omega(n^{1/(d-1)})),$\footnote{Indeed, this is inevitable with all constructions that we are aware of, since they are actually $\AC^0$ circuits and it is known by a result of H\r{a}stad~\cite{Hastad} that any $\AC^0$ circuit of depth $d$ for the Majority function must have size $\exp(\Omega(n^{1/(d-1)})).$} which is superpolynomially larger than the upper bound.

To circumvent this issue, we change the hard functions to the class of \emph{Approximate Majorities}, which is the class of Boolean functions that agree with Majority function on \emph{most} inputs. While this has the downside that we no longer are dealing with an explicitly defined function, the advantage is that the polynomial approximation method of Razborov yields tight lower bounds for some functions from this class. 

Indeed, since the method of Razborov is based on polynomial approximations, it immediately follows that the same proof technique also yields the same lower bound for computing Approximate Majorities. Formally, any $\AC^0[\oplus]$ circuit of depth $d$ computing \emph{any} Approximate Majority must have size $\exp(\Omega(n^{1/2(d-1)}))$. On the upper bound side, it is known from the work of O'Donnell and Wimmer~\cite{odonnellwimmer} and Amano~\cite{Amano-appxmaj} that there \emph{exist} Approximate Majorities that can be computed by monotone $\AC^0$ \emph{formulas} of depth $d$ and size $\exp(O(dn^{1/2(d-1)}))$. (Note that the double exponent $\frac{1}{2(d-1)}$ is now the same in the upper and lower bounds.)

We use the above ideas for our separation between $\AC^0[\oplus]$ formulas and circuits. Plugging in our stronger polynomial approximation for $\AC^0[\oplus]$ formulas, we obtain that any $\AC^0[\oplus]$ formula of depth $d$ computing any Approximate Majority must have size $\exp(\Omega(dn^{1/2(d-1)})).$ In particular, this implies that Amano's construction is tight (up to the universal constant in the exponent) even for $\AC^0[\oplus]$ formulas. 

Further, we also modify Amano's construction~\cite{Amano-appxmaj} to obtain better constant-depth \emph{circuits} for Approximate Majorities: we show that there exist Approximate Majorities that are computed by monotone $\AC^0$ circuits of depth $d$ of size $\exp(O(n^{1/2(d-1)}))$ (the constant in the $O(\cdot)$ is a constant independent of $d$). 

\paragraph{Smaller circuits for Approximate Majority.} Our construction closely follows Amano's, which in turn is related to Valiant's probabilistic construction~\cite{Valiant-maj} of monotone formulas for the Majority function. However, we need to modify the construction in a suitable way that exploits the fact that we are constructing \emph{circuits}. This modification is in a similar spirit to a construction of Hoory, Magen and Pitassi~\cite{hoory2006monotone} who modify Valiant's construction to obtain smaller monotone circuits (of depth $\Theta(\log n)$) for computing the Majority function exactly.

At a high level, the difference between Amano's construction and ours is as follows. Amano constructs random formulas $F_i$ of each depth $i\leq d$ as follows. The formula $F_1$ is the AND of $a_1$ independent and randomly chosen variables. For even (respectively odd) $i>1$, $F_i$ is the OR (respectively AND) of $a_i$ independent and random copies of $F_{i-1}$. For suitable values of $a_1,\ldots,a_d\in \mathbb{N}$, the random formula $F_d$ computes an Approximate Majority with high probability. In our construction, we build a depth $i$ circuit $C_i$ for each $i\leq d$ in a similar way, except that each $C_i$ now has $M$ different outputs. Given such a $C_{i-1}$, we construct $C_i$ by taking $M$ independent randomly chosen subsets $T_1,\ldots,T_M$ of $a_i$ many outputs of $C_{i-1}$ and adding gates that compute either the OR or AND (depending on whether $i$ is even or odd) of the gates in $T_i$. Any of the $M$ final gates of $C_d$ now serves as the output gate. By an analysis similar to Amano's (see also~\cite{hoory2006monotone}) we can show that this computes an Approximate Majority with high probability, which finishes the proof.\footnote{This is a slightly imprecise description of the construction as the final two levels of the circuit are actually defined somewhat differently.}

\section{Preliminaries}

Throughout, $n$ will be a growing parameter. We will consider Boolean functions on $n$ variables, i.e. functions of the form $f:\{0,1\}^n\rightarrow \{0,1\}$. We will sometimes identify $\{0,1\}$ with the field $\F_2$ in the natural way and consider functions $f:\F_2^n \rightarrow \F_2$ instead.

Given a Boolean vector $y\in \{0,1\}^n$, we use $|y|_0$ and $|y|_1$ to denote the number of $0$s and number of $1$s respectively in $y$. 

The Majority function on $n$ variables, denoted $\mathrm{MAJ}_n$ is the Boolean function that maps inputs $x\in \{0,1\}^n$ to $1$ if and only if $|x|_1 > n/2$. 

\begin{definition}
\label{defn:approx-maj}
An $(\varepsilon,n)$-Approximate Majority is a function $f:\{0,1\}^n\rightarrow \{0,1\}$ such that $\prob{x \in \{0,1\}^n}{f(x) \neq \Maj_n(x)} \leq \varepsilon$. 
\end{definition}

As far as we know, the study of this class of functions was initiated by O'Donnell and Wimmer~\cite{odonnellwimmer}. See also~\cite{Amano-appxmaj, BlaisTan}.

We refer the reader to~\cite{AroraBarak, Jukna} for standard definitions of Boolean circuits and formulas. We use $\AC^0$ circuits (respectively formulas) to denote circuits (respectively formulas) of constant depth made up of $\mathrm{AND},\ \mathrm{OR}$ and $\mathrm{NOT}$ gates. Similarly, $\AC^0[\oplus]$ circuits (respectively formulas) will be circuits (respectively formulas) of constant depth made up of $\mathrm{AND},\ \mathrm{OR},\ \mathrm{MOD}_2$ and $\mathrm{NOT}$ gates.

The size of a circuit will denote the number of gates in the circuit and the size of a formula will denote the number of its leaves which is within a constant multiplicative factor of the number of gates in the formula.\footnote{We assume here without loss of generality that the formula does not contain a gate of fan-in $1$ feeding into another.}

\section{Lower Bound}
\label{sec:lbd}

In this section, we show that any $\AC^0[\oplus]$ formulas of depth $d$ computing a $(1/4,n)$-Approximate Majority must have size at least $\exp(\Omega(dn^{1/2(d-1)}))$ for all $d \le O(\log n)$. 

We work over the field $\F_2$ and identify it with $\{0,1\}$ in the natural way. The following concepts are standard in circuit complexity (see, e.g., Beigel's survey~\cite{Beigelpolysurvey}).

\begin{definition}
Fix any $\varepsilon\in [0,1]$. A polynomial $P\in \F_2[X_1,\ldots,X_n]$ is said to be an \emph{$\varepsilon$-approximating polynomial} for a Boolean function $f:\{0,1\}^n\rightarrow \{0,1\}$ if 
\[
\prob{x\in \{0,1\}^n}{f(x) = P(x)}\geq 1-\varepsilon.
\]
\end{definition}

We will use the following result of Smolensky~\cite{Smolensky93} (see also Szegedy's PhD thesis~\cite{Szegedy1989}). 

\begin{lemma}[Smolensky~\cite{Smolensky93}]
\label{lem:razb-maj}
Let $\varepsilon \in (0,\frac{1}{2})$ be any fixed constant. Any $(\frac12-\varepsilon)$-approximating polynomial for the Majority function on $n$ variables must have degree $\Omega(\sqrt{n})$. 
\end{lemma}

\begin{corollary}
\label{cor:approx-maj-deg}
Let $f$ be any $(1/4,n)$-Approximate Majority and $\varepsilon\in (0,1/4)$ an arbitrary constant. Then any $(\frac{1}{4}-\varepsilon)$-approximating polynomial for $f$ must have degree $\Omega(\sqrt{n})$.
\end{corollary}

\begin{proof}
The proof is immediate from Lemma~\ref{lem:razb-maj} and the triangle inequality.
\end{proof}

\begin{definition}
An {\em $\eps$-error probabilistic polynomial of degree $D$} for a Boolean function $f : \{0,1\}^n \to \{0,1\}$ is a random variable $\mb P$ taking values from polynomials in $\F_2[X_1,\ldots,X_n]$ of degree at most $D$ such that for all $x \in \{0,1\}^n$, we have $\Pr[\ f(x) = \mb P(x)\ ] \ge 1 - \eps$.
\end{definition}

\begin{definition}
Let $D_\eps(f)$ be the minimum degree of an $\eps$-error probabilistic polynomial for $f$.
\end{definition}

We will make use of the following two lemmas concerning $D_\eps(\cdot)$.

\begin{lemma}[Razborov~\cite{Razborov}]
\label{lem:razb-OR}
Let $\mathrm{OR}_n$ and $\mathrm{AND}_n$ be the OR and AND functions on $n$ variables respectively. Then $D_\varepsilon(\mathrm{OR}_n), D_\varepsilon(\mathrm{AND}_n)\leq \lceil \log(1/\varepsilon)\rceil$.
\end{lemma}

\begin{lemma}[Kopparty and Srinivasan~\cite{kopparty2012certifying}]\label{la:eps}
There is an absolute constant $c_1$ such that for any $\varepsilon \in (0,1)$, $D_\eps(f) \le c_1\cdot\lceil \log(1/\eps)\rceil \cdot D_{1/8}(f)$ for all Boolean functions $f$.
\end{lemma}

We now state our main result, which shows that every $\AC^0[\oplus]$ formula of size $s$ and depth $d + 1$ admits a $1/8$-error approximating polynomial of degree $O(\frac1d\log s)^d$.

\begin{theorem}
There is an absolute constant $c_2$ such that, if $f$ is computed by an $\AC^0[\oplus]$ formula $F$ of size $s$ and depth $d+1$, then $D_{1/8}(f) \le 3(c_2(\tsfrac{1}{d}\log(s)+1))^d$.
\end{theorem}

\begin{proof}
The proof is an induction on the depth $d$ of the formula. 

The base case $d=0$ corresponds to the case when the formula is a single AND, OR or MOD$_2$  gate and we need to show that $D_{1/8}(f)\leq 3$. In the case that the formula is an AND or OR gate, this follows from Lemma~\ref{lem:razb-OR}. If the formula is a MOD$_2$  gate, this follows from the fact that the MOD$_2$  function is exactly a polynomial of degree $1$. 

Let $d\ge 1$. We assume that the formula $F$ is the AND/OR/MOD$_2$  of sub-formulas $F_1,\ldots,F_m$ computing $f_1,\dots,f_m$ where $F_i$ has size $s_i$ and depth $d+1$. So $F$ has size $s = s_1+\dots+s_m$ and depth $d+2$. Assume that $D_{1/8}(f_i) \le 3(c_2(\frac{1}{d}\log(s_i)+1))^d$ for all $i$. We must show that $D_{1/8}(f) \le 3(c_2(\frac{1}{d+1}\log(s)+1))^{d+1}$.

By Lemma \ref{la:eps}, each $f_i$ has an $s_i/(16s)$-error probabilistic polynomial $\mb P_i$ of degree $c_1\cdot\lceil\log(16s/s_i)\rceil\cdot D_{1/8}(f_i)$, which is at most
\[
  3c_1 \cdot 5(\log(s/s_i)+1)\cdot (c_2(\tsfrac{1}{d}\log(s_i)+1))^d.
\]
Then $(\mb P_1,\dots,\mb P_m)$ jointly computes $(f_1,\dots,f_m)$ with error $1/16$ ($= \sum_{i=1}^m (s_i/(16 s))$). 

By a reasoning identical to the base case, it follows that there exists a $1/16$-error probabilistic polynomial $\mb Q$ of degree $4$ for the output gate of the formula. 

Then $\mb Q(\mb P_1,\dots,\mb P_m)$ is a $1/8$-error probabilistic polynomial for $f$ of degree
\[
  60c_1\cdot\max_i (\log(s/s_i)+1)\cdot (c_2(\tsfrac{1}{d}\log(s_i)+1))^d.
\]

So long as $c_2 \ge 20c_1$, it suffices to show that for all $i$,
\[
  (\log(s/s_i)+1) \cdot (\tsfrac{1}{d}\log(s_i)+1)^d
  \le
  (\tsfrac{1}{d+1}\log(s)+1)^{d+1}.
\]
Consider any $i$ and let $a,b \ge 0$ such that $s_i = 2^a$ and $s = 2^{a+b}$. We must show
\[
  (b+1)\left(\frac{a\vphantom{b}}{d}+1\right)^d
  \le
  \left(\frac{a+b}{d+1}+1\right)^{d+1}.
\]
For fixed $a \ge 0$, as a polynomial in $b$, the function
\[
  p_{a,d}(b) := \left(\frac{a+b}{d+1}+1\right)^{d+1} - (b+1)\left(\frac{a\vphantom{b}}{d}+1\right)^d
\]
is nonnegative over $b \ge 0$ with a unique root at $b = a/d$. This follows from
\[
  \frac{\partial}{\partial b} p_{a,d}(b) = 
  \left(\frac{a+b}{d+1}+1\right)^d - \left(\frac{a\vphantom{b}}{d}+1\right)^d,
\]
which is zero iff $b = a/d$; this value is a minimum of $p_{a,d}$ with $p_{a,d}(a/d) = 0$.
\end{proof}

\begin{corollary}
\label{cor:formula-lbd}
Fix any constant $d$ and let $n\in \mathbb{N}$ be a growing parameter. 
Let $f$ be any $(1/4,n)$-Approximate Majority. Then any $\AC^0[\oplus]$ formula of depth $d$ computing $f$ must have size $\exp(\Omega(dn^{1/2(d-1)}))$ for all $d \le O(\log n)$, where asymptotic notation $O(\cdot)$ and $\Omega(\cdot)$ hide absolute constants (independent of $d$ and $n$). 
\end{corollary}

\begin{proof}
Say that $F$ is an $\AC^0[\oplus]$ formula of depth $d$ and size $s$ computing $f$. Then, by Lemma~\ref{la:eps}, we see that $F$ has a $1/8$-error probabistic polynomial $\mb P$ of degree $D\leq O(O(\frac{1}{d}\log s+1)^{d-1})$. In particular, by an averaging argument, there is some fixed polynomial $P\in \F_2[X_1,\ldots,X_n]$ of degree at most $D$ such that $P$ is a $1/8$-error approximating polynomial for $f$. 

Corollary~\ref{cor:approx-maj-deg} implies that the degree of $P$ must be $\Omega(\sqrt{n})$. Hence, we obtain $O(\frac{1}{d}\log s+1)^{d-1} \ge \Omega(\sqrt{n})$. It follows that
\[
  s \ge \exp(\Omega(dn^{1/2(d-1)}) - O(d)).
\]
Observe that $\Omega(dn^{1/2(d-1)})$ dominates $O(d)$ so long as $d \le \eps \log n$ for some absolute constant $\eps > 0$ (depending on the constants in $\Omega(\cdot)$ and $O(\cdot)$). Hence, we get the claimed lower bound $s \ge \exp(\Omega(dn^{1/2(d-1)}))$ for all $d \le \eps\log n$.
\end{proof}

\section{Upper Bound}

In this section, we show that for any constant $\varepsilon$, there are $(\varepsilon,n)$-Approximate Majorities that can be computed by depth $d$ $\AC^0$ circuits of size $\exp(O(n^{1/2(d-1)}))$.

Let $\varepsilon_0\in (0,1)$ be a small enough constant so that the following inequalities hold for any $\beta \leq \varepsilon_0$
\begin{itemize}
\item $\exp(-\beta) \leq 1-\beta\exp(-\beta)$,
\item $1-\beta \geq \exp(-\beta-\beta^2)\geq \exp(-2\beta)$.
\end{itemize}
(It suffices to take $\varepsilon_0 = 1/2$.) 

We need the following technical lemma.

\begin{lemma}
\label{lem:technicalORAND}
Let $A,s$ be positive reals, $M,n\in\mathbb{N}$,  and $\gamma\in (\frac{1}{n},\frac{1}{10})$ be such that $e^A \geq n^3$, $n \geq \frac{1}{\varepsilon_0}$, and $s\leq n$. Define $I_0(\gamma) := \{y\in \{0,1\}^M\ |\ |y|_1\leq Me^{-A}(1-\gamma)\}$ and $I_1(\gamma) := \{y\in \{0,1\}^M\ |\ |y|_1\geq Me^{-A}(1+\gamma)\}$. If we choose $S\subseteq [M]$ of size $t := \lceil e^{A}\cdot s\rceil$ by picking $t$ random elements from $M$ with replacement, then
\begin{align*}
x\in I_0(\gamma) &\Rightarrow \prob{S}{\bigvee_{j\in S}x_j = 0} \geq \exp(-s)\cdot\exp(s\gamma/2), \\
x\in I_1(\gamma) &\Rightarrow \prob{S}{\bigvee_{j\in S}x_j = 0} \leq \exp(-s)\cdot \exp(-s\gamma).
\end{align*}
Further, if $s\gamma \leq \varepsilon_0$, then the above probabilities can be lower bounded and upper  bounded by $\exp(-s)\cdot (1+s\gamma\exp(-s\gamma))$ and $\exp(-s)\cdot (1-s\gamma\exp(-s\gamma))$ respectively.

A similar statement can be obtained above for the sets $J_1(\gamma) := \{y\in \{0,1\}^M\ |\ |y|_0\leq Me^{-A}(1-\gamma)\}$ and $J_0(\gamma) := \{y\in \{0,1\}^M\ |\ |y|_0\geq Me^{-A}(1+\gamma)\}$, with the event ``$\bigvee_{j\in S}x_j = 0$'' being replaced by the event ``$\bigwedge_{j\in S}x_j = 1$''.
\end{lemma}

\begin{proof}
We give the proof only for $I_0(\gamma)$ and $I_1(\gamma)$. The proof for $J_0(\gamma)$ and $J_1(\gamma)$ is similar.

Consider first the case that $x\in I_1(\gamma)$. In this case, we have the following computation. 
\begin{align}
\prob{S}{\bigvee_{j\in S} x_j = 0} &\leq \left(1-\frac{1+\gamma}{e^A}\right)^{e^A\cdot s}\notag\\
&\leq \exp(-(1+\gamma)\cdot s) \leq \exp(-s)\cdot \exp(-s\gamma).\label{eq:tech2}
\end{align}
The above implies the first upper bound on $\prob{S}{\bigvee_{j\in S} x_j = 0}$ from the lemma statement. When $s\gamma\leq \varepsilon_0$, we further have $\exp(-s\gamma) \leq 1-s\gamma \exp(-s\gamma)$, which implies the second upper bound. This proves the lemma when $x\in I_1$.

Now consider the case that $x\in I_0(\gamma)$. We have
\begin{align}
\prob{S}{\bigvee_{j\in S}x_j = 0} &\geq \left(1-\frac{1-\gamma}{e^A}\right)^{e^A\cdot s + 1}\notag\\
&\geq \exp\left((-\frac{1-\gamma}{e^A} - \frac{1}{e^{2A}})\cdot (e^A\cdot s+1)\right)\notag\\
&= \exp\left(-s+s\gamma -\frac{s}{e^A}-\frac{1-\gamma}{e^A} - \frac{1}{e^{2A}}\right)\notag\\
&\geq \exp\left(-s+s\gamma -\frac{2s}{e^A}\right)\notag\\
&= \exp(-s)\cdot \exp(s\gamma(1-\frac{2e^{-A}}{\gamma})))\label{eq:tech1}
\end{align}
where for the second inequality we have used the fact that since $e^{-A} \leq \frac{1}{n^3}\leq \varepsilon_0$, we have $1- \frac{1-\gamma}{e^A}\geq \exp(-\frac{1-\gamma}{e^A} - \frac{1}{e^{2A}})$. Since $e^{-A}\leq \frac{1}{n^3}\leq \frac{1}{4n}\leq \gamma/4$, we can lower bound the right hand side of (\ref{eq:tech1}) by $\exp(-s)\cdot \exp(s\gamma/2)$. Also, note that 
\begin{align*}
1-\frac{2e^{-A}}{\gamma} &\geq 1-\frac{2/n^3}{1/n}= 1-\frac{2}{n^2}\\ 
&\geq \exp(-1/n) \geq \exp(-\gamma)\geq \exp(-s\gamma).
\end{align*}
This implies that the RHS of (\ref{eq:tech1}) can also be lower bounded by $\exp(-s)\exp(s\gamma\exp(-s\gamma)) \geq \exp(-s)\cdot(1+s\gamma\exp(-s\gamma))$, which implies the claim about $\prob{S}{\bigvee_{j\in S} x_j = 0}$ assuming that $x\in I_0(\gamma)$.
\end{proof}

We now prove the main result of this section. 

\begin{theorem}
\label{thm:amano-ckt}
For any growing parameter $n\in \mathbb{N}$ and $2 \le d \le O(\frac{\log n}{\log\log n})$ and $\eps > 0$, there is an $(\varepsilon,n)$-Approximate Majority $f_n$ computable by a monotone $\AC^0$ circuit with at most $\exp(O(n^{1/2(d-1)}\log(1/\varepsilon)/\varepsilon))$ many gates, where both $O(\cdot)$'s hide absolute constants (independent of $d,\varepsilon$).
\end{theorem}

\begin{proof}
We assume throughout that $\varepsilon$ is a small enough constant and that $n$ is large enough for various inequalities to hold. We will actually construct a monotone circuit of depth $d$ and size $\exp(O(n^{1/2(d-1)}\log(1/\varepsilon)/\varepsilon))$ computing a $(4\varepsilon,n)$-Approximate Majority, which also implies the theorem. 

Fix parameters $A=\lfloor n^{1/2(d-1)}\rfloor$ and $M = \lceil e^{10A}\rceil$. We assume that $A\geq 10\log n$ (which holds as long as $d \le \frac{c\log n}{\log\log n}$ for an absolute constant $c > 0$) and that $\varepsilon \leq \varepsilon_0$.

Define a sequence of real numbers $\gamma_0,\gamma_1,\ldots,\gamma_{d-2}$ as follows:

\begin{align*}
\gamma_0 &= \frac{\varepsilon}{\sqrt{n}}\\
\gamma_i &= A\gamma_{i-1}\exp(-2A\gamma_{i-1}), \text{ for each $i\in [d-2]$.}
\end{align*}

It is clear that $\gamma_i \leq A^i \gamma_0$ for each $i\in [d-2]$. As a result we also obtain 
\begin{align}
\gamma_i &= A^i \gamma_0\exp(-2A(\gamma_0+\gamma_1+\cdots+\gamma_{i-1}))\notag\\
&\geq A^i \gamma_0\exp(-2\gamma_0A(1+A+A^2+\cdots+A^{i-1})) \geq A^i\gamma_0\exp(-3A^i\gamma_0).\label{eq:gamma-i}
\end{align}

Let
\begin{align*}
Y_\varepsilon &= \left\{x\in \{0,1\}^n\ \middle|\ |x|_1\geq \left(\frac{1}{2}+\frac{\varepsilon}{\sqrt{n}}\right)n\right\},\\
N_\varepsilon &= \left\{x\in \{0,1\}^n\ \middle|\ |x|_1\leq \left(\frac{1}{2}-\frac{\varepsilon}{\sqrt{n}}\right)n\right\}.
\end{align*}

The idea is to define a sequence of circuits $C_1,C_2,\ldots,C_{d-2}$ with $n$ inputs and $M$ outputs such that $C_i$ has depth $i$ and $iM$ many (non-input) gates. Further, for odd $i$ 
\begin{align}
x\in N_\varepsilon &\Rightarrow C_i(x)\in I_0(\gamma_i)\notag\\
x\in Y_\varepsilon &\Rightarrow C_i(x)\in I_1(\gamma_i)\label{eq:oddi}
\end{align}
and similarly for even $i$
\begin{align}
x\in N_\varepsilon &\Rightarrow C_i(x)\in J_0(\gamma_i)\notag\\
x\in Y_\varepsilon &\Rightarrow C_i(x)\in J_1(\gamma_i).\label{eq:eveni}
\end{align}

After this is done, we will add on top a depth-$2$ circuit that will reject most inputs from $I_0(\gamma_{d-2})$ or $J_0(\gamma_{d-2})$ --- depending on whether $d-2$ is odd or even respectively --- and accept most inputs from $I_1(\gamma_{d-2})$ or $J_1(\gamma_{d-2})$. 

We begin with the construction of $C_1,\ldots,C_{d-2}$ which is done by induction. 

\paragraph{Construction of $C_1$.} The base case of the induction is the construction of $C_1$, which is done as follows. We choose $M$ i.i.d.\ random subsets $T_1,\ldots,T_M\subseteq [n]$ in the following way: for each $i\in [M]$, we sample $A$ random elements of $[n]$ with replacement. Let $b_i^x = \bigwedge_{j\in T_i}x_j$.

If $x\in N_\varepsilon$, then the probability that $b_i^x = 1$ is given by 
\begin{align*}
\prob{}{b_i^x = 1} &\le \left(\frac{1}{2}- \gamma_0\right)^{A}
\leq \frac{1}{2^A} \left(1-2\gamma_0\right)^A \leq \frac{1}{e^A}(1-\gamma_0A)
\end{align*}
where the last inequality follows from the fact that $(1-z)^A \leq (1-zA + \frac{A^2z^2}{2})$.

Let $\delta = 1/n^3$. Note in particular that $2\delta/\gamma_0A \leq \varepsilon_0$ for large enough $n$.

By a Chernoff bound, the probability that $\frac{1}{M}\sum_i b_i^x\geq \frac{1}{e^A}(1-\gamma_0A)(1+\delta)$ is bounded by $\exp(-\Omega(\delta^2M/e^A))\leq \exp(-\Omega(e^{9A}/n^6))\leq \exp(-n)$, since $e^A \geq n^{10}$. Thus, with probability at least $1-\exp(-n)$, we have
\begin{align}
\frac{\sum_i b_i^x}{M}&\leq \frac{1}{e^A}(1-\gamma_0A)(1+\delta)\notag\\
&\leq \frac{1}{e^A}(1-\gamma_0A + \delta) = \frac{1}{e^A}(1-\gamma_0A(1-\frac{\delta}{\gamma_0A}))\notag\\
&\leq \frac{1}{e^A}(1-\gamma_0A\exp(-\frac{2\delta}{\gamma_0A}))\leq \frac{1}{e^A}(1-\gamma_0A\exp(-\gamma_0A))\notag\\
&\leq \frac{1}{e^A}(1-\gamma_1).\label{eq:C1.1}
\end{align}
Above, we have used the fact that $(1-\frac{\delta}{\gamma_0A})\geq \exp(\frac{-2\delta}{\gamma_0A})$ since $\delta/\gamma_0A \leq \varepsilon_0$ for large enough $n$, as noted above.

If $x\in Y_\varepsilon$, then the probability that $b_i^x = 1$ is given by 
\begin{align*}
\prob{}{b_i^x = 1} &\geq \left(\frac{1}{2}+ \gamma_0\right)^{A}\\
&\geq \frac{1}{2^A} \left(1+2\gamma_0\right)^A \geq \frac{1}{e^A}(1+\gamma_0A)\\
&\geq \frac{1}{e^A}(1+\gamma_0A).
\end{align*}

As above, we can argue that the probability that $\frac{1}{M}\sum_i b_i^x \leq \frac{1}{e^A}(1+\gamma_0A)(1-\delta)$ is at most $\exp(-n)$. Thus, with probability $1-\exp(-n)$

\begin{align}
\frac{\sum_i b_i^x}{M}&\geq \frac{1}{e^A}(1+\gamma_0A)(1-\delta)\notag\\
&\geq \frac{1}{e^A}(1+\gamma_0A -2\delta) = \frac{1}{e^A}(1+\gamma_0A(1-\frac{2\delta}{\gamma_0A}))\notag\\
&\geq \frac{1}{e^A}(1+\gamma_0A\exp(-\frac{4\delta}{\gamma_0A}))\geq \frac{1}{e^A}(1+\gamma_0A\exp(-\gamma_0A))\notag\\
&\geq \frac{1}{e^A}(1+\gamma_1).\label{eq:C1.2}
\end{align}
Thus, by a union bound over $x$, we can fix a choice of $T_1,\ldots,T_M$ so that (\ref{eq:C1.1}) holds for all $x\in N_\varepsilon$ and (\ref{eq:C1.2}) holds for all $x\in Y_\varepsilon$. Hence, (\ref{eq:oddi}) holds for $i=1$ as required. This concludes the construction of $C_1$, which just outputs the values of $\bigwedge_{j\in T_i}x_j$   for each $i$.
 
\paragraph{Construction of $C_{i+1}$.} For the inductive case, we proceed as follows. We assume that $i$ is odd (the case that $i$ is even is similar). So by the inductive hypothesis, we know that (\ref{eq:oddi}) holds and hence that $C_{i}(x) \in I_0(\gamma_i)$ or $I_1(\gamma_i)$ depending on whether $x\in N_\varepsilon$ or $Y_\varepsilon$. Let $\gamma := \gamma_i$. Let the output gates of $C_i$ be $g_1,\ldots,g_M$.

We choose $T_1,\ldots,T_M\subseteq [M]$ randomly as in the statement of Lemma~\ref{lem:technicalORAND} with $s = A$. Note that the chosen parameters satisfy all the hypotheses of Lemma~\ref{lem:technicalORAND}. Further we also have $s\gamma \leq A\cdot A^i\gamma_0 \leq A^{d-1}\cdot \frac{\varepsilon}{\sqrt{n}}\leq \varepsilon_0$.

The random circuit $C'$ is defined to be the circuit obtained by adding $M$ OR gates to $C_i$ such that the $j$th OR gate computes $\bigvee_{k\in T_j}g_k$. Let $b_j^x$ be the output of the $j$th OR gate on $C_i(x)$. 

By Lemma~\ref{lem:technicalORAND}, we have
\begin{align}
x\in N_\varepsilon &\Rightarrow \prob{S}{b_j^x = 0} \geq \exp(-A)\cdot (1+A\gamma\exp(-A\gamma))\notag \\
x\in Y_\varepsilon &\Rightarrow \prob{S}{b_j^x = 0} \leq \exp(-A)\cdot(1-A\gamma\exp(-A\gamma))\label{eq:Ci.1}
\end{align}

Let $\delta = \frac{1}{n^3}$. Note that $A\gamma \in [\frac{1}{\sqrt{n}},\frac{1}{n^{1/2(d-1)}}]$ and hence for large enough $n$, $\frac{2\delta}{A\gamma\exp(-A\gamma)}\leq \varepsilon_0$.

Assume $x\in N_\varepsilon$. In this case, the Chernoff bound implies that the probability that $\sum_{j\in [M]}b_j^x\leq M\exp(-A)\cdot (1+A\gamma\exp(-A\gamma))(1-\delta)$ is at most $\exp(-\Omega(\delta^2M/e^A))\leq \exp(-n)$. When this event does not occur, we have
\begin{align}
\frac{\sum_i b_i^x}{M}&\geq \frac{1}{e^A}(1+A\gamma\exp(-A\gamma))(1-\delta)\notag\\
&\geq \frac{1}{e^A}(1+A\gamma\exp(-A\gamma) -2\delta) = \frac{1}{e^A}(1+A\gamma\exp(-A\gamma)(1-\frac{2\delta}{A\gamma\exp(-A\gamma)}))\notag\\
&\geq \frac{1}{e^A}(1+A\gamma\exp(-A\gamma)\cdot \exp(-\frac{4\delta}{A\gamma\exp(-A\gamma)}))\notag\\
&\geq \frac{1}{e^A}(1+A\gamma\exp(-A\gamma)\cdot \exp(-A\gamma))\notag\\
&\geq \frac{1}{e^A}(1+A\gamma\exp(-2A\gamma))\geq \frac{1}{e^A}(1+\gamma_{i+1}).\label{eq:Ci.2}
\end{align}

We have used above that for large enough $n$, $\frac{2\delta}{A\gamma\exp(-A\gamma)}\leq \varepsilon_0$ and hence $1-\frac{2\delta}{A\gamma\exp(-A\gamma)}\geq \exp(\frac{-4\delta}{A\gamma\exp(-A\gamma)})$.

Similarly when $x\in Y_\varepsilon$, the Chernoff bound tells us that the probability that $\sum_{j\in [M]}b_j^x\geq M\exp(-A)\cdot (1-A\gamma\exp(-A\gamma))(1+\delta)$ is at most $\exp(-n)$. In this case, we get
\begin{align}
\frac{\sum_i b_i^x}{M}&\leq \frac{1}{e^A}(1-A\gamma\exp(-A\gamma))(1+\delta)\notag\\
&\leq \frac{1}{e^A}(1-A\gamma\exp(-A\gamma) +\delta) = \frac{1}{e^A}(1-A\gamma\exp(-A\gamma)(1-\frac{\delta}{A\gamma\exp(-A\gamma)}))\notag\\
&\leq \frac{1}{e^A}(1-A\gamma\exp(-A\gamma)\cdot\exp(-\frac{2\delta}{A\gamma\exp(-A\gamma)}))\notag\\
&\leq \frac{1}{e^A}(1-A\gamma\exp(-A\gamma)\cdot \exp(-A\gamma))\notag\\
&= \frac{1}{e^A}(1-A\gamma\exp(-2A\gamma))\geq \frac{1}{e^A}(1-\gamma_{i+1}).\label{eq:Ci.3}
\end{align}

By a union bound, we can fix $T_1,\ldots,T_M$ so that (\ref{eq:Ci.2}) and (\ref{eq:Ci.3}) are true for all $x\in N_\varepsilon$ and $x\in Y_\varepsilon$ respectively. This gives us the circuit $C_{i+1}$ which satisfies all the required properties.

\paragraph{The top two levels of the circuit.} At the end of the above procedure we have a circuit $C_{d-2}$ of depth $d-2$ and at most $(d-2)M$ gates that satisfies one of (\ref{eq:oddi}) or (\ref{eq:eveni}) depending on whether $d-2$ is odd or even respectively. We assume that $d-2$ is even (the other case is similar). 

Define $\gamma := \gamma_{d-2}$. Recall from (\ref{eq:gamma-i}) that $\gamma \geq A^{d-2}\gamma_0\exp(-3A^{d-2}\gamma_0)\geq A^{d-2}\gamma_0/2$.

Let $M' = \lceil \exp(\frac{10A\log(1/\varepsilon)}{\varepsilon}+10A)\rceil$. We choose $M'$ many subsets $T_1,\ldots,T_{M'}\subseteq [M]$ i.i.d. so that each $T_j$ is picked as in Lemma~\ref{lem:technicalORAND} with $s = 10A\log(1/\varepsilon)/\varepsilon$. Note that 

\[
s\gamma \geq s\frac{A^{d-2}\gamma_0}{2}= \frac{10A\log(1/\varepsilon)}{\varepsilon}\cdot \frac{A^{d-2}}{2}\cdot \frac{\varepsilon}{\sqrt{n}}\geq 5\log(1/\varepsilon).
\]

Say $g_1,\ldots,g_M$ are the output gates of $C_{d-2}$. We define the random circuit $C'$ (with $n$ inputs and $M'$ outputs) to be the circuit obtained by adding $M'$ AND gates such that the $j$th AND gate computes $\bigwedge_{k\in T_j}g_k$. Let $b_j^x$ be the output of the $j$th AND gate on $C_{d-2}(x)$.

By Lemma~\ref{lem:technicalORAND}, we have
\begin{align}
x\in N_\varepsilon &\Rightarrow \prob{S}{b_j^x = 1} \leq \exp(-s)\cdot \exp(-s\gamma)\leq \varepsilon^2\cdot\exp(-s) \notag \\
x\in Y_\varepsilon &\Rightarrow \prob{S}{b_j^x = 1} \geq \exp(-s)\cdot\exp(s\gamma/2)\geq \frac{\exp(-s)}{\varepsilon^2}. \label{eq:Ctop.1}
\end{align}

Say $x\in N_\varepsilon$. By a Chernoff bound, the probability that $\sum_{j}b_j^x \geq 2\varepsilon^2 M'\exp(-s)$ is at most $\exp(-\Omega(\varepsilon^2 M'\exp(-s))) \leq \exp(-\Omega(\varepsilon^2 e^{10A}))\leq \exp(-n)$. Similarly, when $x\in Y_\varepsilon$, the probability that $\sum_{j}b_j^x \leq\frac{M'\exp(-s)}{2\varepsilon^2}$ is also bounded by $\exp(-n)$. By a union bound, we can fix a $T_1,\ldots,T_{M'}$ to get a circuit $C_{d-1}$ such that 
\begin{align}
x\in N_\varepsilon &\Rightarrow |C_{d-1}(x)|_1 \leq 2\varepsilon^2 \exp(-s)M' \notag \\
x\in Y_\varepsilon &\Rightarrow |C_{d-1}(x)|_1 \geq \frac{1}{2\varepsilon^2} \exp(-s)M'. \label{eq:Ctop.2}
\end{align}

This gives us the depth $d-1$ circuit $C_{d-1}$. Note that $C_{d-1}$ has $M' + O(dM) = O(M')$ gates.

To get the depth $d$ circuit, we choose a random subset $T\subseteq [M']$ by sampling exactly $\lceil\exp(s)\rceil$  many elements of $[M']$ with replacement. We construct a random depth-$d$ circuit $C'_d$ by taking the OR of the the output gates of $C_{d-1}$ indexed by the subset $T$. 

From (\ref{eq:Ctop.2}) it follows that 
\begin{align*}
x\in N_\varepsilon &\Rightarrow \prob{T}{C'_d(x) = 1}\leq |T|\cdot 2\varepsilon^2\exp(-s) \leq 4\varepsilon^2 < \varepsilon\\
x\in Y_\varepsilon &\Rightarrow \prob{T}{C'_d(x) = 0}\leq \left(1-\frac{\exp(-s)}{2\varepsilon^2}\right)^{\exp(s)} \leq \exp(-1/2\varepsilon^2) < \varepsilon.
\end{align*}

The final inequalities in each case above hold as long as $\varepsilon$ is a small enough constant. 

It follows from the above that there is a choice for $T$ such that $C'_d$ makes an error --- i.e. $C'_d(x) = 1$ for $x\in N_\varepsilon$ or $C'_d(x) = 0$ for $x\in Y_\varepsilon$ --- on at most a $2\varepsilon$ fraction of inputs from  $N_\varepsilon\cup Y_\varepsilon$. We fix such a choice for $T$ and the corresponding circuit $C$. 

We have 
\begin{align*}
\prob{x\in \{0,1\}^n}{C(x) \neq \mathrm{Maj}_n(x)}&\leq \prob{x\in Y_\varepsilon\cup N_\varepsilon}{C(x) \neq \mathrm{Maj}_n(x)} + \prob{x\in\{0,1\}^n}{x\not\in Y_\varepsilon\cup N_\varepsilon}\\
&\leq 2\varepsilon + \prob{x\in\{0,1\}^n}{x\not\in Y_\varepsilon\cup N_\varepsilon}.
\end{align*}

Finally by Stirling's approximation we get
\[
\prob{x\in\{0,1\}^n}{x\not\in Y_\varepsilon\cup N_\varepsilon} = \frac{1}{2^n}\sum_{m\in [\frac{n}{2}-\varepsilon\sqrt{n}, \frac{n}{2}+\varepsilon\sqrt{n}]}\binom{n}{m} \leq \frac{1}{2^n}\sum_{m\in [\frac{n}{2}-\varepsilon\sqrt{n}, \frac{n}{2}+\varepsilon\sqrt{n}]}\binom{n}{n/2} \leq 2\varepsilon.
\]

Hence we see that the circuit $C$ computes a $(4\varepsilon,n)$-Approximate Majority, which proves Theorem~\ref{thm:amano-ckt}.

The circuit has depth $d$ and size $O(M') = \exp(O(n^{1/2(d-1)}\log(1/\varepsilon)/\varepsilon))$.
\end{proof}

\section{Conclusion}
\label{sec:conclusion}

Our main results extend straightforwardly to $\AC^0[\mathrm{MOD}_p]$ for any fixed prime $p$. The proofs are exactly the same except for the fact that the approximating polynomials of degree $O(\frac{1}{d}\log s)^{d-1}$ from Section~\ref{sec:lbd} are constructed over $\F_p$. 

Using the fact~\cite{Smolensky87} that any $(1/4)$-approximating polynomial over $\F_p$ ($p$ odd) for the Parity function on $n$ variables must have degree $\Omega(\sqrt{n})$, we see that any polynomial-sized $\AC^0[\mathrm{MOD}_p]$ formula computing the Parity function on $n$ variables must have depth $\Omega(\log n)$. This strengthens a result of Rossman~\cite{rossman2015average} which gives this statement for $\AC^0$ formulas.

\paragraph{Acknowledgements.}

We thank Rahul Santhanam for valuable discussions. We also thank the organizers of the 2016 Complexity Semester at St.\ Petersburg, where this collaboration began.

\end{document}